\documentclass[12pt]{article}
\usepackage[left=2cm,right=2cm, top=2.5cm,bottom=2.5cm,bindingoffset=0cm]{geometry}
\usepackage{tikz, titlesec, tikz-cd}
\usetikzlibrary{calc,intersections,through,backgrounds}
\usepackage{amsmath, amssymb}
\usepackage{amsthm}
\usepackage{mathtools}
\usepackage{booktabs, makecell}
\usepackage{adjustbox}

\usepackage{varwidth}
\usepackage{parskip}

\mathtoolsset{showonlyrefs}

\usepackage{array, hyperref}

\titleformat{\section}{\large\bfseries\filcenter}{\thesection}{1em}{}
\titleformat{\subsection}{\bfseries}{\thesubsection}{1em}{}
\titleformat{\subsubsection}[runin]{\bfseries}{\thesubsubsection}{1em}{}[.]

\usepackage{multirow}

\usetikzlibrary{positioning}
\usetikzlibrary{decorations.text}
\usetikzlibrary{decorations.pathmorphing}
\usetikzlibrary{decorations.markings}

\usetikzlibrary{arrows} 
\tikzset{
    >=stealth',
    pil/.style={
           ->,
           ,
           shorten <=2pt,
           shorten >=2pt,}
}

\newcommand{\Z}{\mathbb{Z}}

\newcommand{\R}{\mathbb{R}}

\newcommand{\N}{\mathbb{\Z}_{\geq 0}}
\newcommand{\Id}{\mathrm{Id}}

\newcommand{\RP}{\mathbb{R}\mathrm{P}}
\newcommand{\GL}{\mathrm{GL}}
\newcommand{\PGL}{\mathrm{PGL}}
\newcommand{\D}{\mathcal D}

\newcommand{\Sh}{T}

\newcommand{\gl}{\Mat}

\newcommand{\eps}{{\varepsilon}}
\newcommand{\ep}{\eps}
\newcommand{\Ker}[1]{\mathrm{Ker}\,#1}
\newcommand{\g}{\mathfrak{G}}
\newcommand{\ord}[1]{\mathrm{ord}\,#1}
\newcommand{\sgn}{\mathrm{sgn}}
\newcommand{\Mat}{\mathrm{Mat}}
\newcommand{\ttimes}{\,\tilde \times \,}
\newcommand{\Lax}{\mathcal{L}}
\newcommand{\DO}[2]{\mathrm{DO}_{#1}(#2)}

\newcommand{\ALLDO}[1]{\mathrm{DO}_{#1}}
\newcommand{\IDO}[2]{\mathrm{IDO}_{#1}(#2)}
\newcommand{\PBDO}[2]{\mathrm{PBDO}_{#1}(#2)}
\newcommand{\PSIDO}[1]{\Psi\mathrm{DO}_{#1}}
\newcommand{\IPSIDO}[1]{\mathrm{I}\Psi\mathrm{DO}_{#1}}
\newcommand{\Tr}[1]{\mathrm{Tr}\,#1}

\newcommand{\CPM}[2]{  P(#1, #2)}

\newcommand{\Ad}{{\mathrm{Ad}}}
\newcommand{\grad}{{\mathrm{grad}\,}}
\newcommand{\Proj}{\mathbb P}
\newcommand{\pos}{{>0}}
\renewcommand{\neg}{{<0}}

\newtheorem{lemma}{Lemma}[section]
\newtheorem{proposition}[lemma]{Proposition}
\newtheorem{theorem}[lemma]{Theorem}
\newtheorem{corollary}[lemma]{Corollary}

\theoremstyle{definition}
\newtheorem{remark}[lemma]{Remark}
\newtheorem{definition}[lemma]{Definition}

\newtheorem{example}[lemma]{Example}

\makeatletter
\renewenvironment{proof}[1][\proofname] {\par\pushQED{\qed}\normalfont\topsep6\p@\@plus6\p@\relax\trivlist\item[\hskip\labelsep\bfseries#1\@addpunct{.}]\ignorespaces}{\popQED\endtrivlist\@endpefalse}
\makeatother

\usetikzlibrary{arrows} 
\tikzset{
    >=stealth',
    pil/.style={
           ->,
           ,
           shorten <=2pt,
           shorten >=2pt,}
}
\tikzset{->-/.style={decoration={
  markings,
  mark=at position .7 with {\arrow{>}}},postaction={decorate}}}
  \tikzset{a/.style={decoration={
  markings,
  mark=at position .52 with {\arrow{angle 90}}},postaction={decorate}}}
\tikzset{-<-/.style={decoration={
  markings,
  mark=at position .4 with {\arrow{<}}},postaction={decorate}}}  

\newcounter{ai}

\newcounter{bk}

\title{Long-diagonal pentagram maps} 

\author{Anton Izosimov\thanks{
Department of Mathematics,
University of Arizona;
e-mail: {\tt izosimov@math.arizona.edu}
} \,
and Boris Khesin\thanks{
Department of Mathematics,
University of Toronto;
e-mail: \tt{khesin@math.toronto.edu}
} }

\date{}

\begin{document}
\maketitle
\vspace*{-0.7cm}
\abstract{
The pentagram map on polygons in the projective plane was introduced by R.~Schwartz in 1992 and is by now 
one of the most popular and classical discrete integrable systems.
In the present paper we introduce and prove integrability of long-diagonal pentagram maps on polygons   in $\RP^d$, {by now the most universal pentagram-type map} encompassing all
known integrable cases. We also establish an equivalence of  long-diagonal and bi-diagonal maps
and present a simple self-contained construction of the Lax form for both. 
Finally, we prove the continuous limit of all these maps is  equivalent to the $ (2,d+1)$-KdV equation, generalizing the Boussinesq equation for $d=2$.
}


\section{Introduction}

The pentagram map is a discrete dynamical system on the space of planar polygons introduced in \cite{schwartz1992pentagram}. 
The definition of the pentagram map is illustrated in Figure~\ref{Fig1}: the image of the polygon $P$ under the pentagram map is the polygon $P'$ whose vertices are the intersection points of consecutive shortest diagonals of~$P$, i.e.  diagonals connecting second-nearest vertices.
In~\cite{ovsienko2010pentagram, soloviev2013integrability} it was shown that the pentagram map is a completely integrable system.  As of now, it is one of the most famous discrete integrable systems which features numerous connections to various areas of mathematics.  In particular, the pentagram map has an interpretation in terms of cluster algebras~\cite{GLICK20111019}, networks of surfaces \cite{Gekhtman2016}, the dimer model~\cite{affolter2019vector}, T-systems \cite{kedem2015t}, and Poisson-Lie groups \cite{fock2014loop}. Apparently the most surprising connection is to hydrodynamics: the pentagram map can be viewed as a space-time discretization of the Boussinesq equation \cite{ovsienko2010pentagram}, a shallow water approximation.

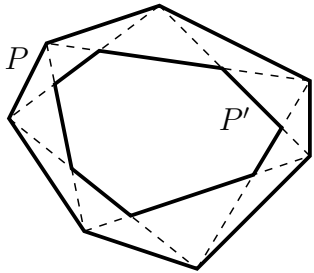
\begin{figure}[ht]
\centering
\begin{tikzpicture}[]
\coordinate (VK7) at (0,0);
\coordinate (VK6) at (1.5,-0.5);
\coordinate (VK5) at (3,1);
\coordinate (VK4) at (3,2);
\coordinate (VK3) at (1,3);
\coordinate (VK2) at (-0.5,2.5);
\coordinate (VK1) at (-1,1.5);

\draw  [line width=0.5mm]  (VK7) -- (VK6) -- (VK5) -- (VK4) -- (VK3) -- (VK2) -- (VK1) -- cycle;
\draw [dashed, line width=0.2mm, name path=AC] (VK7) -- (VK5);
\draw [dashed,line width=0.2mm, name path=BD] (VK6) -- (VK4);
\draw [dashed,line width=0.2mm, name path=CE] (VK5) -- (VK3);
\draw [dashed,line width=0.2mm, name path=DF] (VK4) -- (VK2);
\draw [dashed,line width=0.2mm, name path=EG] (VK3) -- (VK1);
\draw [dashed,line width=0.2mm, name path=FA] (VK2) -- (VK7);
\draw [dashed,line width=0.2mm, name path=GB] (VK1) -- (VK6);

\path [name intersections={of=AC and BD,by=Bp}];
\path [name intersections={of=BD and CE,by=Cp}];
\path [name intersections={of=CE and DF,by=Dp}];
\path [name intersections={of=DF and EG,by=Ep}];
\path [name intersections={of=EG and FA,by=Fp}];
\path [name intersections={of=FA and GB,by=Gp}];
\path [name intersections={of=GB and AC,by=Ap}];

\draw  [line width=0.5mm]  (Ap) -- (Bp) -- (Cp) -- (Dp) -- (Ep) -- (Fp) -- (Gp) -- cycle;

\node at (-0.9,2.3) () {$P$};
\node at (2,1.5) () {$P'$};

\end{tikzpicture}
\caption{The pentagram map.}\label{Fig1}
\end{figure}

The pentagram map admits several different integrable generalizations to multidimensional spaces, { which so far were somewhat sporadic findings. In the present paper we describe} a general construction of integrable pentagram maps in $\RP^d$ which includes all known examples as particular cases. We define the corresponding maps in two different ways: as \textit{bi-diagonal} maps based on intersection of two planes of complementary dimensions, and as \textit{long-diagonal maps}, based on intersection of several hyperplanes. The equivalence of those definitions will be proved later on in the paper.

\begin{definition}\label{def2}
Let $A_+, A_- \subset \Z$ be two disjoint finite non-empty arithmetic progressions with the same common difference, each containing at least two elements\footnote{Otherwise the constructed maps are trivial. Similarly, we exclude consecutive progressions in Definition \ref{def1}  for the same reason.}. Let also $d_\pm := |A_\pm| - 1$, and let $d:= d_+ + d_-$. Then the \textit{bi-diagonal map} associated with progressions $A_\pm$ is a self-map of the space of polygons in $\RP^d$, defined as follows. The image of the polygon $(v_k \in \RP^d)$ is a new polygon $( \tilde v_k \in \RP^d)$ whose vertices are defined by 
\begin{equation*}
\tilde v_k = \Pi_{k,+} \cap \Pi_{k,-}\,,
\end{equation*}
where the planes $\Pi_{k,\pm}$ of complementary dimensions $d_\pm$ are given by $\Pi_{k,\pm} = \mathrm{span}\, \langle v_{j+k} \mid j \in A_\pm \rangle$.
\end{definition}

\begin{example}\label{ex:standpent}
The standard pentagram map  in $\RP^2$ is a bi-diagonal map corresponding to  $A_+=\{0, 2\}$, $A_- =  \{1,3\} $.
\end{example}

\begin{example}\label{ex:p3ex}
Let $A_+ := \{1,4,7\}$ and $A_- := \{3,6\}$.
Then the corresponding bi-diagonal pentagram map in  $\RP^3$ is
\begin{align*}
  \tilde v_k=
\mathrm{span}\, \langle v_{k+1}, v_{k+4}, v_{k+7}\rangle \cap \mathrm{span}\,  \langle v_{k+3}, v_{k+6}\rangle,
\end{align*}
see Figure \ref{Fig2}.
\begin{figure}[t]
\centering
\begin{tikzpicture}[, scale = 1.2]
\coordinate (VK10) at (0.3,0.1);
\coordinate (VK9) at (0.8,-0.3);
\filldraw (VK9) circle (0.7pt);
\coordinate (VK8) at (1.5,-0.5);
\coordinate (VK7) at (2.6,-0.2);
\filldraw (VK7) circle (0.7pt);
\coordinate (VK6) at (2.8,1);
\filldraw (VK6) circle (0.7pt);
\coordinate (VK5) at (2.4,2.1);
\coordinate (VK4) at (1.5,2.8);
\filldraw (VK4) circle (0.7pt);
\coordinate (VK3) at (0.5,3);
\filldraw (VK3) circle (0.7pt);
\coordinate (VK2) at (-0.6,2.3);
\coordinate (VK1) at (-1,1.2);
\filldraw (VK1) circle (0.7pt);
\coordinate (VK) at (-0.5,0);
\filldraw (VK) circle (0.7pt);
\coordinate (VKM1) at (0,-0.5);
%
\path [name path=L17] (VK7) -- (VK1);

\draw  [line width=0.4mm]  (VK10) -- (VK9) -- (VK8) -- (VK7) -- (VK6) -- (VK5) -- (VK4) -- (VK3) -- (VK2) -- (VK1) -- (VK) -- (VKM1);

\path [name path=L36] (VK3) -- (VK6);
\path [name path=L410] (VK4) -- (VK10);
\path [name path=L47] (VK4) -- (VK7);
\path [name path=L38] (VK3) -- (VK8);
\path [name path=L06] (VK) -- (VK6);
\path [name path=L69] (VK9) -- (VK6);
\path [name path=L26] (VK6) -- (VK2);
\path [name path=L39] (VK9) -- (VK3);
\path [name intersections={of=L36 and L410,by=IP}];
\path [name intersections={of=L36 and L47,by=IP2}];
\path [name intersections={of=L38 and L06,by=IP3}];
\path [name intersections={of=L47 and L06,by=IP4}];
\path [name intersections={of=L47 and L69,by=IP5}];
\path [name intersections={of=L17 and L69,by=IP6}];
\path [name intersections={of=L17 and L39,by=IP7}];
\path [name intersections={of=L26 and L39,by=IP8}];
\draw (VK3) -- (IP);
\draw (VK6) -- (IP2);
\draw [dashed] (VK6) -- (IP5);
\draw [dashed] (IP6)-- (VK9) -- (IP7);
\draw [dashed] (IP8) -- (VK3);
\draw [dashed] (IP3) -- (VK) -- (VK3);
\draw [dashed] (IP4) -- (VK6);
\draw [fill=black, fill opacity=0.1] (VK1) -- (VK4) -- (VK7) -- cycle;
\filldraw (IP) circle (1pt);

\node[label={[shift={(-0.2,-0.1)}]right:\small$v_{k+7}$}] at (VK7) () {};
\node[label={[shift={(-0.2,0)}]right:\small$v_{k+6}$}] at (VK6) () {};
\node[label={[shift={(-0.2,0.15)}]right:\small$v_{k+4}$}] at (VK4) () {};
\node[label={[shift={(0,-0.15)}]above:\small$v_{k+3}$}] at (VK3) () {};
\node[label={[shift={(0.15,0)}]left:\small$v_{k+1}$}] at (VK1) () {};
\node[label={[shift={(0.15,0)}]left:\small$v_k$}] at (VK) () {};
\node[label={[shift={(0,0.2)}]below:\small$v_{k+9}$}] at (VK9) () {};
\node[label={[shift={(0,0.2)}]below:\small$\tilde v_k$}] at (IP) () {};

\end{tikzpicture}
\caption{Integrable bi-diagonal/long-diagonal pentagram map in $\RP^3$.
It is a bi-diagonal map corresponding to
 progressions $A_+ := \{1,4,7\}$ and $A_- := \{3,6\}$ and long-diagonal map corresponding to progressions
$B_-:=\{0,3\}$ and $B_+:=\{1\}$.
 }\label{Fig2}
 \end{figure}
\end{example}
A similar definition of pentagram maps based on two progressions appears in \cite{izosimov2018pentagram}. 
However, the maps of  \cite{izosimov2018pentagram} are defined for polygons satisfying certain additional coplanarity assumptions (see {the proof of Theorem \ref{thm:integr}} below), while the above definition is valid for any sufficiently generic polygon  in $\RP^d$.

\medskip

We will also need an alternative definition of bi-diagonal maps in terms of intersection of hyperplanes. For a polygon $(v_k)_{k \in \Z}$ in $\RP^d$ we define its \textit{$m$-diagonal hyperplanes} to be 
$$
\Pi_k := \mathrm{span}\langle v_k, v_{k+m}, \dots, v_{k+m(d-1)} \rangle.
$$
We say that two arithmetic progressions with common difference $m$ are \textit{non-consecutive} if their union is not an arithmetic progression with common difference $m$ (in particular, the two progressions are non-empty and disjoint).

\begin{definition}\label{def1} 
Given a $d$-element set $B:=B_+\cup B_-$ that is the union of two non-consecutive  arithmetic progressions $B_\pm \subset \Z$ with a common difference $m$,   
we define the {\it long-diagonal pentagram map}   on polygons
in $\RP^d$  by intersecting $m$-diagonal hyperplanes spaced out by $B$, i.e.
$$
\tilde v_k:=\bigcap\limits_{j\in B=B_+\cup B_-} \Pi_{j+k}.
$$
\end{definition}
\begin{example}[cf. Example \ref{ex:standpent}]
In this setting the standard pentagram map  in $\RP^2$ corresponds to $m=2$ and $B=\{0\}\cup \{1\} $.
\end{example}

\begin{example}[cf. Example \ref{ex:p3ex}]\label{rp3ex} 
The pentagram map in Figure \ref{Fig2} is an example of a long-diagonal map on polygons in $\RP^3$
with $m=3$, $B_-=\{0,3\}$, $B_+=\{1\}$, and $B=B_-\cup B_+=\{0,1,3\}$,  the union 
 of two arithmetic progressions with common difference $3$.
 \end{example}


We prove equivalence of Definitions~\ref{def2} and \ref{def1} in Section \ref{sect:long}: every bi-diagonal map coincides with a certain long-diagonal map, and vise versa.
\medskip

A much larger class of multidimensional pentagram maps  in $\RP^{d}$ was defined in ~\cite{khesin2013, khesin2016}
using the language of jump $J$ and intersection $I$ tuples integers. These ordered tuples specify which vertices of the polygon form a diagonal hyperplane and which diagonals to intersect to define the map. 
Although general $(J,I)$ maps are apparently not integrable, certain particular cases are. Two integrable series of $(J,I)$ maps, the so-called \textit{short-diagonal} and \textit{dented} maps are described in~\cite{khesin2013, khesin2016}. 
As can be seen from Table \ref{tab:pentmaps}, those maps are particular examples of our general construction (see also Example \ref{ex:more} below).  As a matter of fact, all previously known integrable pentagram maps are either bi-diagonal (equivalently, long-diagonal) maps or their restrictions  to suitable submanifolds. For example, the corrugated maps of \cite{Gekhtman2016} can be viewed as restrictions of dented maps of \cite{khesin2016}.


\medskip

Our main result is integrability of 
\textit{all} long-diagonal pentagram maps,  as well as a description of their continuous limit.
In particular,  we present a self-contained construction of the Lax form for bi-diagonal and long-diagonal maps. 
{ This can be regarded as a unification under one roof of all the known so far integrable cases of pentagram maps, which used to be sufficiently rare discoveries.}
\medskip


\begin{table}
$$
\scriptsize{
\left.\begin{array}{c|c|c|c|c|c}   &&&&& \\[-1em] \textbf{m} & \mathbf{A_+} & \mathbf{A_-}  &  \mathbf{B_+} & \mathbf{B_-} & \mbox{\textbf{The corresponding map}}\\[-1em] &&&&& \\ \hline 2 & \{0,2\} & \{1,3\} & \{0\} & \{1\} & \mbox{Classical pentagram map} \\\hline 1 &\{0,1\} & \{2,3\} & \{0\}& \{2\}&\mbox{Inverse pentagram map} \\\hline 1 & \{0, \dots, k\} & \{k+1, \dots, d+1\}  & \{ k-d+1, \dots, 0\}& \{2, \dots, k+1\}& \mbox{Inverse dented map in $\RP^{d}$}  \\\hline 
2 &\{0,2, \dots, 2k\} & \{1,3, \dots, 2k+1\}  &\{2-2k, 4-2k, \dots, 0\}&\{3-2k, 5-2k, \dots, 1\}& \mbox{Short-diagonal map in $\RP^{2k}$   } 
\\\hline 
2 &\{0,2, \dots, 2k\} & \{1,3, \dots, 2k-1\}  &\{4-2k, 6-2k, \dots, 0\}&\{3-2k, 5-2k, \dots, 1\}& \mbox{Short-diagonal map in $\RP^{2k-1}$    } \\

 \end{array}\right.
 }
$$
\caption{Examples of bi-diagonal/long-diagonal pentagram maps.}\label{tab:pentmaps}
\end{table}

\medskip

We prove integrability of long-diagonal maps on the space of twisted polygons. Recall that a \textit{twisted $n$-gon} in $\RP^d$ is a sequence of points $(v_i \in \RP^d)_{i \in \Z}$ such that  for every $i \in \Z$ we have $v_{i+n} = \phi(v_i)$ where $\phi$ is a projective transformation, known as the \textit{monodromy}. Clearly, long-diagonal pentagram maps take the space of twisted $n$-gons to itself and preserve the monodromy. 


\begin{theorem}\label{thm:integr}
All long-diagonal pentagram maps are completely integrable discrete dynamical systems on the space of projective equivalence classes of twisted $n$-gons in
$\RP^d$. Namely, each of those maps admits a Lax representation with spectral parameter and an invariant Poisson structure such that the spectral invariants of the Lax operator Poisson commute.
\end{theorem}
We prove this theorem in Section \ref{sect:int} based on Definition \ref{def2}. 

\medskip

Our second result is the description of the continuous limit for long-diagonal maps.

\begin{theorem}\label{thm:cont0}
The continuous limit of any long-diagonal pentagram map in dimension $d$ is the $(2, d+1)$-KdV flow of the Adler-Gelfand-Dickey hierarchy on the circle.
\end{theorem}

This theorem is proved in Section \ref{sec:cont}.

\par
\bigskip

{\bf Acknowledgments.}  A.I. was supported by NSF grant DMS-2008021.
B.K. was partially supported by a Simons Fellowship and an NSERC Discovery Grant.

%




%


%

\section{Integrability of bi-diagonal pentagram maps}\label{sect:int}
As above, let $A_+$, $A_-$ be two disjoint finite arithmetic progressions with the same common difference. Let also $d_\pm := |A_\pm| - 1$, $d:= d_+ + d_-$, and $A:= A_+ \cup A_-$.
Consider  a twisted $n$-gon $( v_i \in \RP^d )$ in $\RP^d$, and let $(V_i \in \R^{d+1})$ be its arbitrary quasi-periodic lift to $\R^{d+1}$ (quasi-periodicity of $V_i$ means that $V_{i+n} = MV_i$ for any $i \in \Z$ and a certain non-degenerate matrix $M$). Then, since $|A| = d+2$, the vectors $ V_{j+k}$, $ j \in A$, are linearly dependent for any $k \in \Z$:
\begin{equation}\label{eq:lindep}
\sum_{j \in A} a_{k}^{j} V_{j+k} = 0.
\end{equation}
Note that since the sequence $V_i$ is quasi-periodic, the coefficients $a_{k}^{j}$ are $n$-periodic in the index $k$. Therefore, relations \eqref{eq:lindep} can be equivalently written as
$
\D V = 0,
$
where $V$ is a bi-infinite sequence with entries $V_k$, and $\D$ is an $n$-periodic {difference operator}
\begin{equation}\label{DO}
\D := \sum_{j \in A} a^{j} \Sh^j.
\end{equation}
Here $\Sh$ is the left shift operator  on bi-infinite sequences, $(TV)_k := V_{k+1}$, while the sequences $a^{j}$  of real numbers act on sequences of vectors by term-wise multiplication: $(a^jV)_k :=  a_{k}^{j} V_k$.
Thus, one can encode polygons by {difference operators}. There is, however, more than one operator corresponding to a given polygon:

{

\begin{proposition}\label{newProp}
The $n$-periodic difference operator $\D$ associated to a polygon is unique up to a transformation
\begin{equation}\label{eq:lr}
\D \mapsto \alpha \D \beta^{-1}\end{equation}
where $\alpha, \beta$ are sequences of non-zero real numbers which are quasi-periodic with the same monodromy:
$
{\alpha_{i+n}}/{\alpha_i} = {\beta_{i+n}}/{\beta_i} = z
$
for all $i \in \Z$ and some $z \in \R^*$ independent of $i$.
\end{proposition}
\begin{proof}
Suppose $\D, \D'$ are two operators corresponding to the same polygon $( v_i \in \RP^d )$. Then, by construction, there exist quasi-periodic lifts $(V_i \in \R^{d+1})$, $(V_i' \in \R^{d+1})$ of $(v_i)$ such that
$
\D V = \D' V' = 0.
$
Note that since $(V_i)$, $(V_i')$ are lifts of the same polygon, we have $V_i' = \beta_i V_i$ for some sequence $\beta_i$ of non-zero real numbers. Furthermore, since both lifts are quasi-periodic, the sequence $\beta_i$ is quasi-periodic as well: $\beta_{i+n} = z \beta_i$ for  all $i \in \Z$ and some $z \in \R^*$ independent of $i$. Now we have
$
\D V = \D'\beta V = 0.
$
This forces the relation
$
\D'\beta = \alpha \D
$
for some sequence $\alpha_i$ of non-zero real numbers. So we indeed have
$
\D' = \alpha \D \beta^{-1}.
$
Furthermore, since $\D, \D'$ are both periodic, we must have that $\alpha$ is quasi-periodic, with the same monodromy as~$\beta$.
\end{proof}

Since there are several (infinitely many) difference operators corresponding to a given polygon,  any mapping of the space of polygons to itself lifts to difference operators as a correspondence (a multivalued map) rather than a map}. 
To explicitly describe this correspondence for the case of the bi-diagonal pentagram map, we split the difference operator $\D$ into two parts:
$$
\D_+ := \sum_{j \in A_+} a^{j} \Sh^j, \quad
\D_- := \sum_{j \in A_-} a^{j} \Sh^j.
$$
\begin{theorem}\label{thm0}
The bi-diagonal pentagram map, written in terms of difference operators, is a multivalued map $$\D =  \D_+ + \D_- \quad \longmapsto \quad \tilde \D =  \tilde \D_+ + \tilde \D_-$$ determined by the relation
\begin{equation}\label{mainRelation}
\tilde{\D}_+  {\D}_- = \tilde {\D}_-  {\D}_+.
\end{equation}
In other words, if a polygon $(\tilde v_k)$ is the image of a polygon $(v_k)$ under a bi-diagonal pentagram map, then the associated difference operators $\D$ and $\tilde \D$ can be chosen to satisfy \eqref{mainRelation}.
\end{theorem}
\begin{proof}
Let $V$ be an arbitrary lift of a polygon $(v_k)$ and let $\D$ be a difference operator of the form \eqref{DO} such that $\D V = 0$. Then, by definition of the operator $\D_+$, projections of the entries of the sequence $\D_+V$ to $\RP^d$ belong to $\Pi_{k,+} = \mathrm{span}\langle v_{j+k} \mid j \in A_+ \rangle$. Likewise, projections of the entries of $\D_-V$ to $\RP^d$ are in $\Pi_{k, -} = \mathrm{span}\langle v_{j+k} \mid j \in A_- \rangle$. But $\D V = 0$ implies $\D_- V = -\D_+ V$, so the points in $\RP^d$ defined by $\D_\pm V$ are precisely the vertices of the image  $(\tilde v_k)$ of the polygon $(v_k)$ under the bi-diagonal pentagram map. 
So, as an operator $\tilde \D$ corresponding to the polygon $(\tilde v_k)$ one can take an operator  of the form~\eqref{DO} satisfying 
\begin{equation}\label{anotherRelation}
\tilde \D \D_+ V = 0.
\end{equation}
We will now show that relation \eqref{anotherRelation} implies \eqref{mainRelation}. Indeed, in view of $\D V = 0$,  relation \eqref{anotherRelation} can be rewritten as
\begin{equation}
(\tilde \D \D_+ - \tilde \D_+  \D )V = 0.
\end{equation}
But
$$
\tilde \D \D_+ - \tilde \D_+  \D = \tilde \D_- \D_+ - \tilde \D_+  \D_-,
$$
so
\begin{equation}\label{linDep}
(\tilde \D_- \D_+ - \tilde \D_+  \D_-)V = 0.
\end{equation}
Furthermore, one has $$\tilde \D_- \D_+ - \tilde \D_+  \D_-= \!\!\!\sum_{j \in A_+ + A_-} \!\!\! b^{j} \Sh^j,$$ where $b^j$ are certain $n$-periodic sequences, and $A_+ + A_-$ is the Minkowski sum. But for arithmetic progressions with the same common difference one has $$|A_+ + A_-| = |A_+| + |A_-| - 1 = d + 1.$$ So, since $V$ is a generic sequence of vectors in $\R^{d+1}$, and the operator $\tilde \D_- \D_+ - \tilde \D_+  \D_-$ annihilates $V$, that operator must be zero. The result follows. \end{proof}
{
\begin{proof}[Proof of Theorem \ref{thm:integr} ]

We first recall a result of \cite{izosimov2018pentagram} on integrability of pentagram maps on \textit{$A$-corrugated polygons}. Let $A \subset \Z$, $|A| \geq 2$ be a finite set of integers containing at least two elements, and let $D := \max(A) - \min(A) - 1$. Then a polygon $\{ v_i\}$ in $\RP^D$ is called \textit{$A$-corrugated}  if, for any $i \in \Z$, the $|A|$ points $\{ v_{i+j} \mid j \in A\}$ belong to an $|A|-2$ dimensional plane. Pentagram maps on such polygons can be defined when $A \subset \Z$ can be partitioned as $A = A_+ \sqcup A_-$, where $A_\pm \subset \Z$ are finite {arithmetic progressions with the same common difference}.
The pentagram map associated to such a partition is defined by
\begin{equation}\label{jmap}
\tilde v_i : = \mathrm{span}\langle v_{i + j} \mid j \in A_+ \rangle \cap \mathrm{span}\langle v_{i + j} \mid j \in A_-\rangle.
\end{equation}
Using the above construction, one can represent twisted $A$-corrugated polygons by periodic difference operators of the form \eqref{DO}. Moreover, in this case the correspondence between polygons and difference operators up to transformations \eqref{eq:lr} is a bijection  \cite[Proposition 3.3]{izosimov2018pentagram}. Rewritten in terms of difference operators, the pentagram map on $A$-corrugated polygons takes the same form \eqref{mainRelation}. It turns out that that  \eqref{mainRelation} is an integrable relation: it can be written in the Lax form and has an invariant Poisson structure such that the spectral invariants of the Lax operator Poisson commute. So, since bi-diagonal maps are described by the same relation  \eqref{mainRelation} (even though in this case the relation between polygons and operators is not bijective, see Remark \ref{rem:corr} below), they are also integrable. 
Namely, following \cite{izosimov2018pentagram}, consider the operators $\D_\pm$ associated to a polygon as described above. Let  $\Lax:= {\D}_-^{-1}{\D}_+$. Here $ {\D}_-^{-1}$ (and hence $\Lax$) is understood as a \textit{pseudo-difference operator}, i.e. a formal series of the form $\sum_{j =k}^\infty a^{j} \Sh^j$ where $k \in \Z$ is arbitrary, and the coefficients $a^{j}$ are periodic sequences. Note that by Proposition~\ref{newProp}, the operator $\Lax$ is defined up to conjugation by a quasi-periodic sequence.
Now, rewriting~\eqref{mainRelation} as
\begin{equation}\label{refactRelation}
\tilde{\D}_-^{-1}\tilde{\D}_+ = \D_+  \D_-^{-1},
\end{equation}
we see that $\Lax:= {\D}_-^{-1}{\D}_+$ transforms as
\begin{equation}\label{Lax}
\Lax = {\D}_-^{-1}{\D}_+ \mapsto \tilde \Lax = \D_+  \Lax \D_+^{-1},
\end{equation}
thus providing a Lax form for relation  \eqref{mainRelation}. A Poisson structure preserved by this map comes from a Poisson-Lie structure on the group of pseudo-difference operators, while Poisson-commutativity of spectral invariants of $\Lax$ with respect to that  structure is a direct consequence of the Poisson-Lie property, see \cite{izosimov2018pentagram} for details. 
\end{proof}
}
\begin{remark}
The group of invertible pseudo-difference operators is isomorphic to the loop group of invertible $n \times n$ matrices over the field $\R((z))$ of formal Laurent series  \cite[Remark 3.8]{izosimov2018pentagram}. Therefore, \eqref{Lax} can be seen as a Lax representation with spectral parameter.
\end{remark}
\begin{remark}\label{rem:corr}
Note that the map taking a polygon to the associated Lax operator $\Lax:= {\D}_-^{-1}{\D}_+$ (defined up to conjugation by scalar sequences) is not a bijection but (generically) a finite-to-one map. Thus, while operator dynamics \eqref{refactRelation} and equivalent Lax dynamics \eqref{Lax} are genuine integrable systems, the bi-diagonal pentagram dynamics should instead be viewed as a ramified covering of an integrable system  (in particular, the Poisson structure on Lax operators may become singular when lifted to polygons). This covering is best understood geometrically, via a description of  \eqref{refactRelation} as a pentagram map on {$A$-}corrugated polygons {(see the proof of Theorem \ref{thm:integr} above)}.
It is shown in  \cite{izosimov2018pentagram} that this map is equivalent to dynamics  \eqref{mainRelation}. At the same time, bi-diagonal pentagram maps can be viewed as lifts  of~\eqref{mainRelation} to a certain covering space. This covering can be described using a projection. Namely, any projection
$
\RP^{D} \to \RP^d,
$
where $D := \max(A) - \min(A) - 1$ and $d: = |A| - 2$,
intertwines the pentagram dynamics on $A$-corrugated polygons in $\RP^{D}$ with the bi-diagonal pentagram map on generic polygons in $ \RP^{d}$. Note though that a projection of a twisted polygon is, generally speaking, not a twisted polygon, unless we project from a subspace fixed by the monodromy. Since generic monodromies have finitely many invariant subspaces, there are generically finitely many polygons in $\RP^{d}$ corresponding to a given $A$-corrugated polygon in $\RP^{ D} $. In other words, we have proved the following statement.

\begin{proposition}
There is a finite ramified covering
$$
\left\{ {\begin{varwidth}{6cm}\centering projective equivalence classes of generic polygons in $\RP^{d}$\end{varwidth}} \right\} \to \left\{ {\begin{varwidth}{6cm}\centering projective equivalence classes of $A$-corrugated polygons in $ \RP^{D}$\end{varwidth}}.\right\}
$$
This covering intertwines the pentagram map on $A$-corrugated polygons with the corresponding bi-diagonal pentagram map on generic polygons.
\end{proposition}

Note that in terms of difference operators, this covering is just the map sending a polygon to the associated Lax operator  $\Lax:= {\D}_-^{-1}{\D}_+$ {(defined up to conjugation by scalar sequences)}.


%
\end{remark}

\medskip


\section{Equivalence of bi-diagonal and long-diagonal pentagram maps}\label{sect:long}

%
%

We start by recalling the definitions from the introduction in slightly more detail.

\begin{definition} Given a positive integer $m$ and a (twisted) $n$-gon $(v_k)_{k \in \Z}$ in $\RP^d$  its \textit{long-diagonal (or $m$-diagonal) hyperplanes} are 
$$
\Pi_k := \mathrm{span}\langle v_k, v_{k+m}, \dots, v_{k+m(d-1)} \rangle.
$$
Next, fix a $d$-tuple $B:=B_+\cup B_-$ that is the union of two non-consecutive arithmetic progressions $B_\pm$ of integers with a common difference $m$,   $|B|=|B_+|+|B_-|=d$.
Then the {\it long-diagonal pentagram map} is defined as follows: 
the vertex $\tilde v_k$ of the image of the polygon $(v_k)_{k \in \Z}$ under that map is given by 
intersecting $m$-diagonal hyperplanes spaced out by $B$, i.e.
$$
\tilde v_k:=\bigcap\limits_{j\in B=B_+\cup B_-} \Pi_{j+k}.
$$
\end{definition}

\begin{remark} 
As we mentioned above, the   usual pentagram map  in $\RP^2$ corresponds to $m=2$ and $B=\{0\}\cup \{1\} $.

For $m=2$ and $B=\{0\}\cup \{1,3\} $ we recover the pentagram map in $\RP^3$, which is known to be 
numerically integrable, see map \#4 from the table in \cite[Section 6]{khesin2015}. 
Now its integrability is established since it is a long-diagonal pentagram map. 

The pentagram map in Figure \ref{Fig2} is an example 
of a long-diagonal map on polygons in $\RP^3$
with $m=3$ and $B=\{0,3\}\cup \{1\}$. 
\end{remark}

\begin{example}\label{ex:more} 
The  short-diagonal pentagram map  in $\RP^d$ corresponds to $m=2$ and $B=B_+\cup B_- $ with
$B_\pm$ being arithmetic $2$-progressions starting, respectively, at $0$ and $1$, and
of length ``half-$d$". More precisely, for even dimension $d=2\ell$ one sets 
$B_-=\{0,2,...,2\ell-2\}$ and $B_+=\{1,3,..., 2\ell-1\}$ (so that $|B_-|=|B_+|=\ell$), while for odd
$d=2\ell+1$ one sets $B_-=\{0,2,...,2\ell\}$ and $B_+=\{1,..., 2\ell-1\}$ (so that $|B_-|=\ell+1$ and $|B_+|=\ell$). Since in both cases $B=\{0,..,d-1\}$, the above 
definition is equivalent to the more customary definition of the short-diagonal map 
$$
\tilde v_k:=\bigcap\limits_{j=0}^{d-1} \Pi_{j+k}\,,
$$
see  \cite{khesin2013}. Note that the sets $B_\pm$ used to describe the short-diagonal maps in Table \ref{tab:pentmaps} are different from these by a shift, which leads to a shift of indices in the definition of the map.  
\end{example}
We now show that every bi-diagonal map is a long-diagonal map and vice versa. Fix a positive integer $m$ and let
\begin{gather}
\mathcal A_m =\left \{ \parbox{10cm}{\centering pairs of disjoint arithmetic progressions $A_\pm$ with common difference $m$, each containing at least two elements}\right\},\\
\mathcal B_m = \left \{ \parbox{10cm}{\centering pairs of non-consecutive arithmetic progressions $B_\pm$ with common difference $m$}\right\}.
\end{gather}

\begin{proposition}\label{prop:ldtoij}
There is a bijection $\phi \colon \mathcal A_m \to \mathcal B_m$ such that the bi-diagonal map associated with $A_\pm$ coincides with the long-diagonal map associated with $\phi(A_\pm)$.

\end{proposition}

In the proof below we  realize each of the two diagonals given by  $A_\pm$ by a unique appropriate intersection of long-diagonal hyperplanes.

\begin{proof}
Let $\phi(A_\pm) := B_\pm$, where $B_\pm$ are progressions
\begin{equation}\label{eq:bfroma}
\begin{gathered}
B_+ :=  \{\min(A_+) + \min(A_-) - \max(A_-) + m,  \dots, \min(A_+) \}, \\
B_- :=  \{\min(A_-) + \min(A_+) - \max(A_+) + m,  \dots, \min(A_-) \}.
\end{gathered}
\end{equation}
We first prove that the progressions $B_\pm$ are non-consecutive, so $\phi$ is well-defined. Note that $
|B_\pm|  = |A_\mp| -1,
$
so $|A_\pm| \geq 2$ implies $B_\pm \neq \emptyset$.  Also observe that for  a bijection $f \colon \Z \to \Z$ given by \begin{equation}\label{eq:fmap}f(i) = \min(A_+) + \min(A_-) - i\end{equation} we have $f(B_+) \subset A_-$ and $f(B_-) \subset A_+$. So, since $A_+ \cap A_- = \emptyset$, it follows that $B_+ \cap B_- = \emptyset$.  So, since $B_\pm$ are non-empty and disjoint it suffices to show that $\max B_+ + m \neq \min B_-$ and $\max B_- + m \neq \min B_+$. Observe that
\begin{gather}
\max B_+ + m = \min B_-  \quad \Leftrightarrow \quad  \min(A_+) + m = \min(A_-) + \min(A_+) - \max(A_+) + m \\ \Leftrightarrow \quad  \max A_+ = \min A_-.
\end{gather}
So, $A_+ \cap A_- = \emptyset$ implies $\max B_+ + m \neq \min B_-$. Likewise, we get $\max B_- + m \neq \min B_+$. Thus, $B_\pm$ are indeed non-consecutive, as needed.

We now show that $\phi$ is a bijection. From \eqref{eq:bfroma} we find that $A_\pm$ must be of the form
\begin{equation}\label{eq:afromb}
\begin{gathered}
A_+ =  \{\max(B_+),  \dots, \max(B_+) + \max(B_-) - \min B_- + m \}, \\
A_- =  \{\max(B_-),  \dots, \max(B_+) + \max(B_-) - \min B_+ + m \},
\end{gathered}
\end{equation}
so $\phi$ is injective. To prove surjectivity we need to show that progressions \eqref{eq:afromb} are disjoint and satisfy $|A_\pm| \geq 2$. The latter is straightforward since  $|A_\pm|  = |B_\mp| + 1 \geq 2$. To prove the former consider the progressions
$f(B_\pm)$ where $f$ is defined by \eqref{eq:fmap}. Since $B_\pm$ are disjoint, so are $f(B_\pm)$. Therefore, to show that $A_\pm = f(B_\mp) \cup \{ \max(B_+) + \max(B_-) - \min B_\mp + m\}$ are disjoint it suffices to show that \begin{gather}\label{eq:notin1} \max(B_+) + \max(B_-) - \min B_- + m \notin f(B_+), \\ \label{eq:notin2}  \max(B_+) + \max(B_-) - \min B_+ + m \notin f(B_-).\end{gather}
Note that $\max(B_+) + \max(B_-) - \min B_- \in f(B_-)$, so if $\max(B_+) + \max(B_-) - \min B_- + m $ is in $f(B_+)$, then it must be the minimal element of that set. In that case we have
$$
 \max(B_+) + \max(B_-) - \min B_- + m = \min f(B_+) = f(\max B_+) = \max B_-,
$$
so $\max B_+ + m =  \min B_- $, which contradicts $B_\pm$ being nonconsecutive. So, \eqref{eq:notin1} is proved. The proof of \eqref{eq:notin2} is analogous. Therefore, $\phi$ is indeed a bijection.

Finally, we show that for $A_\pm$ and $B_\pm$ related by \eqref{eq:afromb}, the bi-diagonal map defined by $A_\pm$ coincides with the long-diagonal map defined by $B_\pm$.
By definition, the long-diagonal pentagram map associated with $B_\pm$  is given by
$$
\tilde v_k:=\bigcap_{j \in B_+ \cup B_-}\Pi_{j+k}\,,
$$
where $\Pi_k = \mathrm{span}\,\langle v_k, v_{k+m}, \dots, v_{k+ m(d-1)} \rangle$.
Furthermore, we have
\begin{gather}
\bigcap\nolimits_{j \in B_+}\Pi_{j+k} = \bigcap\nolimits_{j \in B_+}  \mathrm{span}\,\langle v_{i + j +k} \mid i = 0, m, \dots, m(d-1) \rangle \\ = \mathrm{span}\,\langle v_{j+k}\mid j = \max(B_+), \max(B_+) + m, \dots, \min(B_+) +  m(d-1)\rangle  \\ = \mathrm{span}\,\langle v_{j+k}\mid j   \in A_+\rangle.
\end{gather}
Likewise,
$$
\bigcap_{j \in B_-}\Pi_{j+ k} = \mathrm{span}\,\langle v_{j+k}\mid j   \in A_-\rangle.
$$
So indeed the long-diagonal pentagram map determined by $B_\pm$ coincides with the bi-diagonal pentagram map
associated with~$A_\pm$.
\end{proof}

\begin{corollary} \label{cor2}
All long-diagonal pentagram maps are completely integrable discrete dynamical systems on the space of projective equivalence classes of twisted $n$-gons in $\RP^d$. 
\end{corollary}

The above definition of long-diagonal maps (along with their duals, defined below), as well as their restrictions to subclasses of polygons, generalizes all previously known examples of integrable pentagram maps \cite{ovsienko2010pentagram, Gekhtman2016, khesin2013, khesin2016}: short-diagonal, dented, 
deep dented, corrugated and partially corrugated ones.

\medskip

We recall the above mentioned duality in a somewhat more general setting, which we will also need for the 
continuous limit discussed below.
\medskip

\begin{definition}[cf. \cite{khesin2016}]
 Let $P = (p_1, \dots, p_d)$ be a strictly monotonic (increasing or decreasing) $d$-tuple of  integers (a \textit{jump sequence}). Given a (twisted) $n$-gon $(v_k)_{k \in \Z}$ in $\RP^d$, 
 its \textit{diagonal hyperplanes associated with $P$} are defined by
$$
\Pi_k := \mathrm{span}\langle v_{k+p_1}, ..., \dots, v_{k +p_d} \rangle.
$$
Let also $Q= (q_1, \dots, q_d)$  be another monotonic $d$-tuple of  integers (an \textit{intersection sequence}). Then the pentagram map $T_{P,Q}$ is defined as follows: the vertex $\tilde v_k$ of the image of the polygon $(v_k)_{k \in \Z}$ under $T_{P,Q}$ is given by
$$
\tilde v_k:=\Pi_{k+q_1}\cap  ...\cap \Pi_{k+q_d}.
$$
Modulo a shift of indices, the map $T_{P,Q}$ is defined by the mutual differences between elements of $P$ and of $Q$,
and this is how such maps have been usually defined. 
Namely, let $P$ and $Q$ be strictly increasing $d$-tuples. 

Now we set  $J = (j_1, \dots, j_{d - 1})$ to be a $(d-1)$-tuple of positive integers (a \textit{jump tuple})
defined as $j_l:=p_{l+1}-p_l$ for $l=1,...,d-1.$ In other words, the corresponding $P$-diagonal hyperplane is 
$$
\Pi_k := \mathrm{span}\langle v_k, v_{k+j_1}, \dots, v_{k + j_1+ \dots + j_{d-1}} \rangle.
$$
Similarly, the $(d-1)$-tuple  $I= (i_1, \dots, i_{d - 1})$   of positive integers (an \textit{intersection tuple})
corresponds to the $d$-tuple $Q$ as follows: $i_l:=q_{l+1}-q_l$ for $l=1,...,d-1.$ 
Then the pentagram map $T_{I,J}=T_{P,Q}$ (modulo index shifts) is 
$$
\tilde v_k:=\Pi_k\cap \Pi_{k+i_1}\cap...\cap \Pi_{k+i_1+...+i_{d-1}}.
$$
\end{definition}

In the following example we relate these two different systems of notations.

\begin{example}
$a)$ The {\it standard pentagram map}  in $\RP^2$ corresponds to $P=(0,2)$, $Q=(0,1)$, while $J = (2)$ and  $I= (1)$.

$b)$ The {\it long-diagonal pentagram maps} in $\RP^d$ correspond to the jump sequence 
$P=(0,m,\dots, m(d-1))$ and the intersection sequence $Q=B_-\cup B_+$, the union of two arithmetic $m$-sequences. 
Here $J = (m,\dots,m)$, while the $(d-1)$-intersection tuple $I$, standing for differences between consecutive elements of $B_+ \cup B_-$  can assume various forms, e.g.  $I=(m,\dots,m, p, m,\dots,m)$ for any $p\in \Z_\pos$, or 
$I=(m,..,m,q,m-q,q, m-q,m,...m)$, or  $I=(m,..,m,q,m-q,q, m-q,q,m,...m)$ for any positive integer $q<m$, or parts thereof.

$c)$ The case of $m=1$ with $P=(0,1,\dots, d-1)$, $Q=(0,1,\dots,q, q+p, q+p+1, \dots,d+p-2)$, i.e. 
$J=(1,\dots,1)$ and $I=(1,...,1,p,1,...,1)$,  corresponds to {\it dual (deep) dented maps} on generic polygons and can be restricted to produce pentagram maps on corrugated polygons.

$d)$ The {\it short-diagonal pentagram maps}  in $\RP^d$ correspond to $m=2$, $P=(0,2,\dots,2(d-1))$, $Q=(0,1,\dots,d-1)$, while $J = (2, \dots, 2)$ and  $I= (1, \dots, 1)$, see \cite{khesin2013}.

$e)$  For $m=3$ the long-diagonal pentagram map  in $\RP^3$ discussed in Examples  \ref{ex:p3ex} and \ref{rp3ex} (see Figure \ref{Fig2}) has $P=(0,3,6)$ and $Q=(0,1,3)$. The corresponding jump and intersection 2-tuples are respectively 
$ J=(3,3) $ and $I=(1,2)$, the differences between the entries of $P$ and $Q$. The integrability of this map 
is now established by Corollary \ref{cor2}. Note that in \cite{khesin2015} the numerical {\it non-integrability} was observed for a pentagram map $T_{J,I}$ in $\RP^3$ with 
$J=(3,3)$ and $I=(1,1)$. It turns out, however, that for $J=(3,3)$ and $I=(1,2)$ the map becomes integrable!
\end{example}

\begin{definition}
Given $d$-tuples $P$ and $Q$, we say that the pentagram map $T_{Q,P}$ {\it is dual to} $T_{P,Q}$.
Similarly, $T_{I,J}$  {\it is dual to} $T_{J,I}$.
\end{definition}

\begin{remark}\label{rem:duality}
There is an important relation between dual pentagram maps that allows one to interchange $d$-tuples $P$ and $Q$ (and, respectively, $J$ and $I$). Namely, for $P=(p_1,\dots,p_d)$ we 
denote by $P^*$ the $d$-tuple of integers read backwards: $P^*:=(p_d,\dots, p_1)$. Similarly, for 
$J=(j_1,\dots,j_{d-1})$ we set $J^*:=(j_{d-1},\dots,j_1)$, etc. 
Then modulo a shift of indices
$$
T_{P,Q}=T^{-1}_{Q^*, P^*} \quad {\rm and }\quad T_{J,I}=T^{-1}_{I^*, J^*}\,, 
$$
where $T^{-1}$ stands for the inverse pentagram map on twisted polygons in $\RP^d$, see  \cite[Theorem 1.4]{khesin2016}. 
In particular, this implies that the maps $T_{P,Q}$ and $T_{Q^*, P^*}$ are integrable or nonintegrable simultaneously.

Finally, note that the definition of $d$-tuples $P$ and $Q$ in long-diagonal maps allows their 
inversion: the tuples $P^*$ and $Q^*$ have the same form, since arithmetic $m$-progressions and their unions remain such when read backwards.
It follows that, whenever convenient,  to study properties of such maps one may consider the dual long-diagonal maps
$T_{Q,P}$, where the  $d$-tuple $Q$, defining the diagonal hyperplane, is the union of two $m$-progressions, while the intersection
$d$-tuple is a single $m$-progression. We will need this duality to study their continuous limits.
\end{remark}

\medskip


\section{Continuous limit of long-diagonal pentagram maps}\label{sec:cont}

\begin{theorem}\label{thm:contlim}
The continuous limit of all long-diagonal/bi-diagonal maps in $\RP^d$ is equivalent to the $ (2,d+1)$-KdV equation, generalizing the Boussinesq
equation for $d=2$.
\end{theorem}

In the limit as $n\to\infty$  a generic twisted $n$-gon becomes  
a smooth non-degenerate quasi-periodic curve $\gamma(x)$. The limit of 
pentagram maps  is an evolution on such curves constructed as follows. 
Consider the lift of $\gamma(x)$ in $\RP^d$ to a curve 
 $G(x)$ in $\R^{d+1}$  defined by the conditions that the components of the vector function
$G(x)=(G_1,...,G_{d+1})(x)$ provide the homogeneous coordinates for
$\gamma(x)=(G_1:...:G_{d+1})(x)$ in $\RP^d$ and \begin{equation}\label{wronskian}\det(G(x),G^{(1)}(x),...,G^{(d)}(x))=1\end{equation} for all $x\in \R$, where $G^{(i)}$ are consecutive derivatives of the vector-function $G(x)$ with respect to $x$.
Furthermore, $G(x+2\pi)=MG(x)$ for a given $M\in \mathrm{SL}_{d+1}(\R)$. Then $G(x)$
satisfies the linear differential equation of order $d+1$:
$$
G^{(d+1)}+u_{d-1}(x)G^{(d-1)}+...+u_1(x)G^{(1)}+u_0(x)G=0
$$
with periodic coefficients $u_i(x)$, which is a continuous limit of difference equation defining a space $n$-gon. 

Fix a small $\eps>0$ and  let $P$ be any $d$-tuple $P=(p_1,\dots,p_d)$ of increasing integers.
For the $P$-diagonal hyperplane
$$
\Pi_k:= \mathrm{span}\langle v_{k+p_1},  ..., v_{k+p_d} \rangle
$$
its  continuous analogue  is the hyperplane $\Pi_\eps(x)$ passing through
$d$ points $\gamma(x+p_1\eps),...,\gamma(x+p_d\eps)$
of the curve $\gamma$. In what follows we are going to make a parameter shift in $x$
(equivalent to shift of indices) and define
$ \Pi_\eps(x):=\mathrm{span}\langle  \gamma(x+k_1\eps),\gamma(x+k_1\eps),...,\gamma(x+k_{d}\eps)\rangle$, for any real $k_1<k_2<...<k_{d}$  such that $\sum_l k_l=0$.

Let $\zeta_\eps (x)$ be the envelope curve for the family of hyperplanes $\Pi_\eps(x)$  in $\RP^d$ for a fixed $\eps$.
(Geometrically  the envelope can be thought of as the intersection of infinitely close ``consecutive"
 hyperplanes of this family along the curve.)
The envelope condition means that  $\Pi_\eps(x)$ are the osculating hyperplanes of the curve $\zeta_\eps (x)$, that is  the point $\zeta_\eps (x)$ belongs to the hyperplane $\Pi_\eps(x)$, while the vector-derivatives $\zeta^{(1)}_\eps (x),...,\zeta^{(d-1)}_\eps (x)$ span this hyperplane
 for each $x$. It means that the lift of $\zeta_\eps (x)$ in $\RP^d$ to $Z_\eps (x)$ in $\R^{d+1}$
satisfies the system of $d$ equations:
$$
\det ( G(x+k_1\eps), ..., G(x+k_{d}\eps),  Z^{(j)}_\eps(x) )=0,\quad j=0,...,d-1.
$$
Here the lift $Z_\eps (x)$ is again defined by the constraint $\det(Z_\eps(x),Z_\eps'(x),...,Z_\eps^{(d)}(x))=1$
for all  $x\in \R$. One can show that the expansion of the lift $Z_\eps(x)$ has the form
$$
Z_\eps(x)=G(x)+\eps^2 F(x)+{\mathcal O} (\eps^3)\,,
$$
where there is no term linear in $\eps$ due to the condition $\sum_l k_l=0$.

\begin{definition}
A continuous limit of the pentagram map  is  the evolution of the curve $\gamma$ in the direction of the envelope $\zeta_\eps$, as $\eps$ changes: ${d}G/{dt} =F$.
More explicitly, the lift $Z_\eps(x)$  satisfies the family of differential equations:
$$
Z_\eps^{(d+1)}+u_{d-1,\ep}(x)Z_\eps^{(d-1)}+...+u_{1,\ep}(x)Z_\eps^{(1)}+u_{0,\ep}(x)Z_\eps=0, 
$$
where $Z_0(x)=G(x)$, i.e. $ u_{j,0}(x)=u_{j}(x)$.
Then the corresponding expansion of the coefficients $u_{j,\ep}(x)$ as $u_{j,\ep}(x)=u_{j}(x)+\ep^2w_j(x)+{\mathcal O}(\ep^3)$,
defines the continuous limit of the pentagram map as the system of evolution differential equations $du_j(x)/dt\, =w_j(x)$ for $j=0,...,d-1$. 
\end{definition}

This definition of limit via an envelope assumes that we are dealing with consecutive hyperplanes in the pentagram map, i.e. the intersection sequence $Q=(0,1,...,d-1)$ is comprised of consecutive integers. 
The  corresponding intersection tuple is $ I={\mathbf 1}:=(1,...,1)$ or can be taken as its multiple. 
One also obtains the envelope by intersecting hyperplanes with step $\delta$ and taking limit $\delta \to 0$ by formally passing from difference operators to the differential ones, see \cite{nackan2020continuous}.

We start with reminding the following theorem, a variation of a result from \cite{khesin2016} or \cite{nackan2020continuous}, which is the main ingredient of the proof of Theorem \ref{thm:contlim}.

\begin{remark}
Note that while long-diagonal maps are related with equally spaced jump tuple $J=(m,...,m)$ (and an appropriate 
intersection tuple $I$), the envelope 
is related to an  equally spaced intersection tuple $I=(m,...,m)$ (and any $J$). These two cases are dual 
(and hence inverses) of each other, thanks to a key property of the pentagram maps:  $T_{I,J}^{-1}$ coincides with the map $T_{J^*,I^*}$  (modulo shift of indices), see Remark \ref{rem:duality} above and  \cite{khesin2016}.
\end{remark}

\begin{theorem}{\rm (cf. \cite{khesin2016})}\label{thm:cont}
The continuous limit of any generalized pentagram map $T_{ I,J}$ for any $ J=(i_1,...,i_{d-1})$ and 
$ I=m{\mathbf 1}$
(and in particular, of the  inverse of any long-diagonal/bi-diagonal pentagram map) in dimension $d$ defined by the system $du_j(x)/dt\, =w_j(x), \, j=0,...,d-1$ for $x\in S^1$  is the $(2, d+1)$-KdV flow of the Adler-Gelfand-Dickey  hierarchy on the circle.
\end{theorem}

\begin{remark}
Recall that the $(k, d+1)$-KdV flow is defined on  linear differential operators
 $L= \partial^{d+1} + u_{d-1}(x) \partial^{d-1} + u_{d-2}(x) \partial^{d-2} + ...+ u_1(x) \partial + u_0(x)$  of order $d+1$ with periodic coefficients $u_j(x)$, where $\partial^{l}$ stands for $d^l/dx^l$.
One can define the fractional power
$L^{k/{d+1}}$ as a pseudo-differential operator for any positive integer $n$ and take
its pure differential part  $Q_k :=(L^{k/{d+1}})_+$. In particular, for $k=2$ one has 
$$
Q_2= \partial^2 + \dfrac{2}{d+1}u_{d-1}(x). 
$$ 
Then the $(k, d+1)$-KdV equation is the  evolution equation on (the coefficients of) $L$ given by $dL/dt= [Q_k,L] $, see \cite{Adler}. For $k=2$ this gives the $(2, d+1)$-KdV system 
\begin{equation}\label{eq:kdv}
\frac{dL}{dt}= [Q_2,L]:= \left[ \partial^2 + \dfrac{2}{d+1}u_{d-1}(x), L\right].
\end{equation}
 For $d=2$ and $k=2$ the  (2,3)-KdV system gives evolution equations on the coefficients $u$ and $v$ of the operator 
 $L=\partial^3+ u(x) \partial + v(x)$. Upon elimination of $v$ this reduces to the classical Boussinesq equation  on the circle: $$u_{tt}+2(u^2)_{xx}+u_{xxxx}=0,$$ which appears as the continuous limit of the 2D pentagram map \cite{ovsienko2010pentagram}.
\end{remark}

\begin{proof}[Proof of Theorem \ref{thm:cont}]
The proof is based on the  expansion of the envelope $Z_\eps(x)$  in the parameter $\eps$: one can show that 
$$
Z_\eps(x)=G(x)+\eps^2 C_{d,m, I}\left(\partial^2+ \dfrac{2}{d+1}u_{d-1}(x)  \right)G(x)+{\mathcal O} (\eps^3)
$$ 
as $ \eps\to 0$,  for a certain non-zero constant $C_{d,m, I}$. 
This gives the following evolution of the curve $G(x)$ given by the
$\eps^2 $-term of this expansion:
$$\frac{dG}{dt} = \left(\partial^2+ \dfrac{2}{d+1}u_{d-1}\right)G,$$ or which is the same,
${d}G/{dt} =Q_2G$.

To find the evolution of the  differential operator $L$  tracing it recall that for any $t$, 
the curve $G$ and the operator $L$ are related by the differential equation $LG=0$.
In particular,
$d(LG)/dt=({d}L/{dt}) G +  L ({d}G/{dt})=0$, which, in view of ${d}G/{dt} =Q_2G$, implies
$$
\left(\frac{dL}{dt} + LQ_2\right)G = 0,
$$
and hence
$$
\left(\frac{dL}{dt} - [Q_2, L]\right)G = 0,
$$
where $[Q_2, L]:=Q_2L-LQ_2$. But ${dL}/{dt} - [Q_2, L]$ is an operator of order $\leq d$, so it can only annihilate the vector-function $G(x) \in \R^{d+1}$ if $L$ satisfies the $(2,d+1)$-KdV equation 
$$
\frac{dL}{dt} = [Q_2, L].
$$
which proves Theorem \ref{thm:cont}.
\end{proof}

\begin{proof}[Proof of Theorem \ref{thm:contlim}]
By definition, any long-diagonal map has jump tuple $J=(m,...,m)$
and a certain intersection tuple $I$. Now pass to the inverse map. For such a  map the new jump tuple is
$\bar J=I^*$, and the new  intersection tuple is $\bar I=J^*=(m,...m)$. 
According to Theorem \ref{thm:cont} the continuous limit for pentagram maps $T_{\bar I,\bar J}$  in $\RP^d$ with such
$\bar I$ and $\bar J$ (including the inverses of long-diagonal maps) is equivalent to the  $(2,d+1)$-KdV equation.

Finally, note that the inverse of Equation \eqref{eq:kdv} is the same differential equation with the reversed time variable.
Thus the continuous limit of the long-diagonal pentagram maps is given by the same
$(2,d+1)$-KdV equations upon the changing time $t\to -t$, which we treat on equal footing with the original KdV flows. 
\end{proof}

\begin{remark}
Intersections of nonconsecuitive hyperplanes or hyperplanes depending on their indices the continuous limit might not be an envelope and it can be arranged more arbitrarily. For instance, for some special choices, one can obtain higher equations of the KdV hierarchy, see \cite{MB, nackan2020continuous}.
In our case the continuous limit turned out to be a familiar system thanks to the regular structure 
of the diagonal hyperplane $J=(m,...,m)$.
\end{remark}

\begin{remark}
It would be interesting to obtain the $(2,d+1)$-KdV equation as the continuous limit directly from the Lax form \eqref{Lax} for
the pentagram maps, by formally passing to the limit from the linear difference equations defining polygons  to the linear
differential equations defining curves in   $\R^{d+1}$, cf.~\cite{nackan2020continuous}.

\end{remark}

\bibliographystyle{plain}
\bibliography{pent.bib}

\end{document}